\documentclass{mfat}
\pagespan{197}{209}

\usepackage{mmap}
\usepackage[utf8]{inputenc}
\usepackage[all]{xy}

\theoremstyle{plain}
\newtheorem{theorem}{Theorem}
\newtheorem{thm}{\protect\theoremname}

\theoremstyle{remark}
\newtheorem{rem}[thm]{\protect\remarkname}
\newtheorem{definition}{Definition}
\newtheorem{example}[thm]{\protect\examplename}

\makeatother

\providecommand{\examplename}{Example}
\providecommand{\remarkname}{Remark}
\providecommand{\theoremname}{Theorem}

\begin{document}

\title[Fractional statistical dynamics and fractional kinetics]
    {Fractional statistical dynamics and fractional kinetics}

\author{Jos{\'e} Lu{\'\i}s da Silva}
\address{CCM, University of Madeira, Campus da Penteada, 9020-105
Funchal, Portugal}
\email{luis@uma.pt}

\author{Anatoly N. Kochubei}
\address{Institute of Mathematics, National Academy of Sciences of Ukraine,
3 Tereshchenkivs'ka, Kyiv, 01601, Ukraine}
\email{kochubei@imath.kiev.ua}

\author{Yuri Kondratiev}
\address{Department of Mathematics, University of Bielefeld,
 D-33615 Bielefeld, Germany}
\email{kondrat@mathematik.uni-bielefeld.de}

\date{17/03/2016;\ \  Revised  03/04/2016}

\begin{abstract}
We apply the subordination principle to  construct
kinetic fractional statistical dynamics in the continuum in terms of  solutions   to
Vlasov-type
hierarchies.  As a by-product we obtain the evolution
of the density of particles in the fractional kinetics in terms of
a non-linear Vlasov-type
kinetic equation. As an application we study the intermittency of
the fractional mesoscopic dynamics.
\end{abstract}

\maketitle

\section{Introduction}

A general scheme for the study of  Markov dynamics for interacting particle systems (IPS for short) in  the continuum includes the following steps. We start with a heuristic Markov generator  $L$
defined on functions over a configuration space of the system.
Associated with this generator is
a forward Kolmogorov equation  for states of the system (a.k.a.\ a Fokker--Planck equation (FPE)).
A solution to this equation gives the so-called statistical dynamics of the model under
consideration \cite{FKK12}.
A constructive approach to the existence and uniqueness problem for
solution of the FPE and for
the analysis of its properties exploits
the possibility of writing
this equation as a hierarchical chain of evolution equations for time dependent correlation functions
\cite{FKK12, FKK15}.  This step corresponds to a microscopic description of the system.

A mesoscopic  level of the study is related with a Vlasov-type scaling limit
for the dynamics that leads to a kinetic or Vlasov hierarchy for correlation functions. This scaling limit destroys the Markov
property of the evolution
of the limiting Vlasov--Fokker--Planck equation (VFPE): for an initial probability
measure the solution, in general, is no longer a measure.
Still, the resulting dynamics has a conditional Markov property in the following
sense. If we  start with a Poisson
initial state then the solution
of the VFPE will be given by a flow of Poisson measures on the configuration space. In
theoretical physics this fact is known as the chaos  propagation  property.

The Poisson flow which appears  in the Vlasov limit is completely characterized by the density
function $\rho_t(x)$, which corresponds to the Poisson measure from the
flow at time $t\geq 0$.
A specific feature of the mesoscopic limit is a non-linear Vlasov-type kinetic equation for this density.  In
most cases this equation may be informally derived directly from the form of the
generator $L$.
However, a rigorous realization of the above scheme is a non-trivial task
for each particular model \cite{FKK10, FKK12, FKKK15}.  The study of the resulting kinetic equations
for concrete  Markov dynamics of interacting particle systems in the continuum
belongs to the general theory of non-local non-linear evolution equations which has
been under active development in recent years.

The aim of the present paper is to extend the concept of statistical dynamics and related structures to the case of fractional time derivatives. From the probabilistic point of view this means that we leave the Markov dynamical framework by introducing a random time change in the corresponding Markov process -- see for example
\cite{Mura_Taqqu_Mainardi_08, Mainardi97}.   In the language of functional
analysis we are no more in the arena of semigroup evolutions.

Below we discuss the concept of a fractional Fokker--Plank equation (FPE) and the
related fractional statistical dynamics, which is still an evolution
in the space of probability measures on the configuration space.
The mesoscopic scaling of the generator of this evolutions leads  to the same result as
for the initial FPE. The latter leads us to the concept of  a fractional VFPE. A subordination principle provides for the representation of
the solution to this equation as a flow of measures that is a transformation
of a Poisson flow for the initial VFPE.
Note that the density function for the fractal kinetics is a subordination of the solution to the initial Vlasov equation. This density characterizes the kinetic behavior of the fractional statistical dynamics, but it is
not the same as the solution to the Vlasov equation with a fractional time  derivative, as is typically assumed in theoretical physics.

In this paper we leave open the problem of rigorous realization of scaling approach for particular models. Instead,
our considerations are focused on questions about the properties of subordinated flows. In particular, we clarify
the possibility of having time-dependent random point processes with an asymptotic intermittency property as a result
of subordination of Poisson flows.

\section{Preliminaries}

Let ${\mathcal{B}}({{\mathbb{R}}^{d}})$ be the family of all Borel
sets in ${{\mathbb{R}}^{d}}$, $d\geq1$ and let ${\mathcal{B}}_{b}({{\mathbb{R}}^{d}})$
denote the system of all bounded sets in ${\mathcal{B}}({{\mathbb{R}}^{d}})$.

The space of $n$-point configurations in an arbitrary
$Y\in{\mathcal{B}}({{\mathbb{R}}^{d}})$
is defined by
\[
\Gamma_{Y}^{(n)}:=\big\{\eta\subset Y\big|\;|\eta|=n\big\},\quad n\in{\mathbb{N}}.
\]
We also set
$\Gamma_{Y}^{(0)}:=\{\emptyset\}$. As a set, $\Gamma_{Y}^{(n)}$
may be identified with the symmetrization of
\[
\widetilde{Y^{n}}=\big\{(x_{1},\ldots,x_{n})\in Y^{n}\big|x_{k}\neq x_{l}\text{ if }k\neq l\big\}.
\]

The configuration space over the space ${{\mathbb{R}}^{d}}$ consists
of all locally finite subsets (confi\-gurations) of ${{\mathbb{R}}^{d}}$,
namely,
\begin{equation}
\Gamma=\Gamma_{{\mathbb{R}}^{d}}:=\big\{\gamma\subset{{\mathbb{R}}^{d}}\big||\gamma\cap\Lambda|<\infty,\text{ for all }\Lambda\in{\mathcal{B}}_{\mathrm{b}}({{\mathbb{R}}^{d}})\big\}.\label{eq:conf_space}
\end{equation}
The space $\Gamma$ is equipped with the vague topology, i.e., the
minimal topology for which all mappings $\Gamma\ni\gamma\mapsto\sum_{x\in\gamma}f(x)\in{\mathbb{R}}$
are continuous for any continuous function $f$ on ${{\mathbb{R}}^{d}}$
with compact support. Note that the summation in $\sum_{x\in\gamma}f(x)$
is taken over only finitely many points of $\gamma$ belonging
to the support of $f$. It was shown in \cite{Kondratiev2006} that
with the vague topology $\Gamma$ may be metrizable and it becomes
a Polish space (i.e., a complete separable metric space). Corresponding
to this topology, the Borel $\sigma$-algebra ${\mathcal{B}}(\Gamma)$
is the smallest $\sigma$-algebra for which all mappings
\[
\Gamma\ni\gamma\mapsto|\gamma_{\Lambda}|\in{\mathbb{N}}_{0}:={\mathbb{N}}\cup\{0\}
\]
 are measurable for any $\Lambda\in{\mathcal{B}}_{b}({{\mathbb{R}}^{d}})$.
Here $\gamma_{\Lambda}:=\gamma\cap\Lambda$, and $|\cdot|$ the cardinality of a finite set.

It follows that one can introduce the corresponding Borel $\sigma$-algebra,
which we denote by ${\mathcal{B}}(\Gamma_{Y}^{(n)})$. The space of
finite configurations in an arbitrary $Y\in{\mathcal{B}}({{\mathbb{R}}^{d}})$
is defined by
\[
\Gamma_{0,Y}:=\bigsqcup_{n\in{\mathbb{N}}_{0}}\Gamma_{Y}^{(n)}.
\]
This space is equipped with the topology of disjoint unions. Therefore
one can introduce the corresponding Borel $\sigma$-algebra ${\mathcal{B}}(\Gamma_{0,Y})$.
In the case of $Y={{\mathbb{R}}^{d}}$ we will omit the index $Y$
in the notation, thus
$\Gamma_{0}:=\Gamma_{0,{{\mathbb{R}}^{d}}}$
$\Gamma^{(n)}:=\Gamma_{{{\mathbb{R}}^{d}}}^{(n)}$.

The restriction of the Lebesgue product measure $(dx)^{n}$ to $\bigl(\Gamma^{(n)},{\mathcal{B}}(\Gamma^{(n)})\bigr)$
will be denoted by $m^{(n)}$, and we set
$m^{(0)}:=\delta_{\{\emptyset\}}$.
The Lebesgue--Poisson measure $\lambda$ on $\Gamma_{0}$ is defined
by
\begin{equation}
\lambda:=\sum_{n=0}^{\infty}\frac{1}{n!}m^{(n)}.\label{eq:lp_measure}
\end{equation}
For any $\Lambda\in{\mathcal{B}}_{b}({{\mathbb{R}}^{d}})$, the restriction
of $\lambda$ to $\Gamma_{\Lambda}:=\Gamma_{0,\Lambda}$ will be also
denoted by $\lambda$. The space $\bigl(\Gamma,{\mathcal{B}}(\Gamma)\bigr)$
is the projective limit
of the family of spaces $\big\{(\Gamma_{\Lambda},{\mathcal{B}}(\Gamma_{\Lambda}))\big\}_{\Lambda\in{\mathcal{B}}_{b}({{\mathbb{R}}^{d}})}$.  The Poisson
measure $\pi$ on $\bigl(\Gamma,{\mathcal{B}}(\Gamma)\bigr)$ is given
as the projective limit of the family of measures
$\{\pi^{\Lambda}\}_{\Lambda\in{\mathcal{B}}_{\mathrm{b}}({{\mathbb{R}}^{d}})}$,
where $\pi^{\Lambda}:=e^{-m(\Lambda)}\lambda$ is the probability
measure on $\bigl(\Gamma_{\Lambda},{\mathcal{B}}(\Gamma_{\Lambda})\bigr)$.
Here $m(\Lambda)$ is the Lebesgue measure of $\Lambda\in{\mathcal{B}}_{b}({{\mathbb{R}}^{d}})$.

For any measurable function $f:{{\mathbb{R}}^{d}}\rightarrow{\mathbb{R}}$
we define a \emph{Lebesgue--Poisson exponent}
\begin{equation}
e_{\lambda}(f,\eta):=\prod_{x\in\eta}f(x),\quad\eta\in\Gamma_{0};\quad e_{\lambda}(f,\emptyset):=1.\label{eq_lp_exponent}
\end{equation}
Then, by (\ref{eq:lp_measure}), for $f\in L^{1}({{\mathbb{R}}^{d}},dx)$
we obtain $e_{\lambda}(f)\in L^{1}(\Gamma_{0},d\lambda)$ and
\begin{equation}
\int_{\Gamma_{0}}e_{\lambda}(f,\eta)\,d\lambda(\eta)=\exp\left(\int_{\mathbb{R}^{d}}f(x)\,dx\right).\label{LP-exp-mean}
\end{equation}

A set $M\in{\mathcal{B}}(\Gamma_{0})$ is called bounded if there
exists $\Lambda\in{\mathcal{B}}_{b}({{\mathbb{R}}^{d}})$ and $N\in{\mathbb{N}}$
such that $M\subset\bigsqcup_{n=0}^{N}\Gamma_{\Lambda}^{(n)}$. We will
make use of the following classes of functions on $\Gamma_{0}$:
(i) $L_{\mathrm{ls}}^{0}(\Gamma_{0})$
is the set of all measurable functions on $\Gamma_{0}$ which have
local support, i.e., $G\in L_{ls}^{0}(\Gamma_{0})$, if there exists
$\Lambda\in\mathcal{B}_{b}(\mathbb{R}^{d})$ such that $G\upharpoonright_{\Gamma_{0}\setminus\Gamma_{\Lambda}}=0$, while
(ii) $B_{bs}(\Gamma_{0})$ is the set of bounded measurable functions with
bounded support, i.e., $G\upharpoonright_{\Gamma_{0}\setminus B}=0$
for some bounded $B\in\mathcal{B}(\Gamma_{0})$.

In fact, any
${\mathcal{B}}(\Gamma_{0})$-measurable function $G$ on $\Gamma_{0}$
is a sequence of functions $\bigl\{ G^{(n)}\bigr\}_{n\in{\mathbb{N}}_{0}}$,
where $G^{(n)}$ is a ${\mathcal{B}}(\Gamma^{(n)})$-measurable function
on $\Gamma^{(n)}$.

On $\Gamma$ we consider the set of cylinder functions $\mathcal{F}_{cyl}(\Gamma)$.
These functions are characterized by the
relation
$F(\gamma)=F\upharpoonright_{\Gamma_{\Lambda}}(\gamma_{\Lambda})$.

The following mapping from $L_{ls}^{0}(\Gamma_{0})$ into
${{\mathcal{F}}_{cyl}}(\Gamma)$
which  plays the key role in our further considerations:
\begin{equation}
KG(\gamma):=\sum_{\eta\Subset\gamma}G(\eta),\quad\gamma\in\Gamma,\label{eq:k-transform}
\end{equation}
where $G\in L_{ls}^{0}(\Gamma_{0})$.  (See, for example,
\cite{KK99}, \cite{Le75a,Le75b}). The summation in
(\ref{eq:k-transform}) is taken over all finite sub-configurations
$\eta\in\Gamma_{0}$ of the (infinite) configuration
$\gamma\in\Gamma$; this relationship is represented symbolically
by $\eta\Subset\gamma$. The mapping $K$ is linear, positivity
preserving, and invertible, with
\begin{equation}
K^{-1}F(\eta):=\sum_{\xi\subset\eta}(-1)^{|\eta\setminus\xi|}F(\xi),\quad\eta\in\Gamma_{0}.\label{k-1trans}
\end{equation}
Here and in the sequel, inclusions like $\xi\subset\eta$ hold for
$\xi=\emptyset$ as well as for $\xi=\eta$. We denote the restriction
of $K$ onto functions on $\Gamma_{0}$ by $K_{0}$.

A measure $\mu\in{\mathcal{M}}_{\mathrm{fm}}^{1}(\Gamma)$ is called
locally absolutely continuous with respect to
(w.r.t.) a
Poisson measure $\pi$ if for any $\Lambda\in{\mathcal{B}}_{b}({{\mathbb{R}}^{d}})$
the projection of $\mu$ onto $\Gamma_{\Lambda}$ is absolutely continuous
w.r.t.~projection of $\pi$ onto $\Gamma_{\Lambda}$. By \cite{KK99},
there exists in this case a \emph{correlation functional} $k_{\mu}:\Gamma_{0}\rightarrow{\mathbb{R}}_{+}$
such that the following equality holds for any $G\in B_{bs}(\Gamma_{0})$:
\begin{equation}
\int_{\Gamma}(KG)(\gamma)\,d\mu(\gamma)=\int_{\Gamma_{0}}G(\eta)k_{\mu}(\eta)\,d\lambda(\eta).\label{eq:corr_funct}
\end{equation}
Restrictions $k_{\mu}^{(n)}$ of this functional on $\Gamma_{0}^{(n)}$,
$n\in{\mathbb{N}}_{0}$, are called \emph{correlation functions} of
the measure $\mu$.
Note that $k_{\mu}^{(0)}=1$.

\section{Mesoscopic statistical dynamics}

\label{sec:general_scheme}

In this section we introduce the general scheme of Vlasov scaling for the Markov dynamics of
IPS -- interacting particle systems --
on configuration space.  Thus we assume
that the initial distribution (the state of particles) in our system is
a probability measure $\mu_{0}\in\mathcal{M}^{1}(\Gamma)$ with corresponding
correlation function $k_{0}=(k_{0}^{(n)})_{n=0}^{\infty}$. The distribution
of particles at time $t>0$ is the measure $\mu_{t}\in\mathcal{M}^{1}(\Gamma)$,
and $k_{t}=(k_{t}^{(n)})_{n=0}^{\infty}$ its correlation function.
If the evolution of states $(\mu_{t})_{t\geq0}$ is determined a priori
by a heuristic Markov generator $L$, then $\mu_{t}$ is the solution
of the forward Kolmogorov equation (or Fokker--Plank equation (FPE)),
\begin{equation}
\begin{cases}
\frac{\partial\mu_{t}}{\partial t} & =L^{*}\mu_{t},\\
\mu_{t}|t_{=0} & =\mu_{0},
\end{cases}\label{eq:FPe}
\end{equation}
where $L^{*}$ is the adjoint operator.
In terms of the time-dependent correlation functions $(k_{t})_{t\geq0}$
corresponding to $(\mu_{t})_{t\geq0}$, the FPE may be rewritten
as an infinite system of evolution equations
\begin{equation}
\begin{cases}
\frac{\partial k_{t}^{(n)}}{\partial t} &
=(L^{\triangle}k_{t})^{(n)},\\ k_{t}^{(n)}|_{t=0} &
=k_{0}^{(n)},\quad n\geq0,
\end{cases}\label{eq:hierarchy}
\end{equation}
where $L^{\triangle}$ is the image of $L^{*}$ in a Fock-type space
of vector-functions $k_{t}=(k_{t}^{(n)})_{n=0}^{\infty}$. In
applications to concrete models, the expression for the operator
$L^{\triangle}$ is obtained from the operator $L$ via combinatorial
calculations (cf.~\cite{KK99}). The following diagram

\[
\xymatrix{L\ar[r]^{\mathrm{duality}}\ar[dd]_{K} & L^{*}\ar[dd]^{K^{*}}\\
\\
\hat{L}=K^{-1}LK & L^{\triangle}=\hat{L}^{*}=K^{*}L^{*}(K^{-1})^{*}
}
\]

\noindent
describes the relationships.

The evolution equation (\ref{eq:hierarchy}) is nothing but a hierarchical
system of equations to the Markov generator $L$. This system is the
analogue of the BBGKY-hierarchy of the Hamiltonian dynamics \cite{Bo62}.

Our interest now turns to Vlasov-type scaling of stochastic dynamics for
the IPS in a continuum. This scaling leads to
so-called kinetic description of the considered model. In  the language of theoretical physics
we are dealing with a mean-field type scaling
which is adopted to preserve the spatial structure. In addition, this scaling will lead
to the limiting hierarchy, which possesses a chaos
propagation property. In other words, if the initial distribution
is Poisson (non-homogeneous) then the time evolution of states will
maintain this property.  We refer to \cite{FKK10} for a
general approach, concrete examples, and additional references.

There exists a standard procedure for deriving Vlasov scaling $L_{V}^{\triangle}$
 from the generator  $L^{\triangle}$
in (\ref{eq:hierarchy}).  Heuristically,  $L_{V}^{\triangle}$ corresponds to a (non-Markov) generator
 $L_{V}$  on observables which may be reconstructed form  $L_{V}^{\triangle}$ just on the level of combinatorial
 calculations.    All together, it  gives us the  following chain of transformed operators:

\[
\xymatrix{L\ar[rr] &  &  L_{V} \ar[rr]  & & L_{V}^{*}\ar[rr]  &  & L_{V}^{\triangle}.}
\]

\noindent
The specific type of scaling is dictated by the model in
question. The process leading from $L^{\triangle}$ to
$L_{V}^{\triangle}$  produces a non-Markovian generator  $L_V$
since it lacks the positivity-preserving property. Therefore
instead of (\ref{eq:FPe}) we consider the following kinetic FPE,
\begin{equation}
\begin{cases}
\frac{\partial\mu_{t}}{\partial t} & =L_{V}^{*}\mu_{t},\\
\mu_{t}|t_{=0} & =\mu_{0},
\end{cases}\label{eq:FPe1}
\end{equation}
and observe that if the initial distribution satisfies
$\mu_{0}=\pi_{\rho_{0}}$,
then the solution is of the same type, i.e., $\mu_{t}=\pi_{\rho_{t}}$.

In terms of correlation functions, the kinetic FPE (\ref{eq:FPe1})
gives rise to the following Vlasov-type hierarchical chain (Vlasov hierarchy)
\begin{equation}
\begin{cases}
\frac{\partial k_{t}^{(n)}}{\partial t} &
=(L_{V}^{\triangle}k_{t})^{(n)},\\ k_{t}^{(n)}|_{t=0} &
=k_{0}^{(n)},\quad n \geq 0.
\end{cases}\label{eq:vlasov_hierarchy}
\end{equation}

\begin{rem}
\label{rem:Leb_Poi_exp}
\begin{enumerate}
\item[]{1.}
In applications it is important to consider the Lebesgue--Poisson
expo\-nents
$k_{0}(\eta)=e_{\lambda}(\rho_{0},\eta)=\prod_{x\in\eta}\rho_{0}(x)$
as the initial condition. The scaling $L_{V}^{\triangle}$ should
be such  that the dynamics $k_{0}\mapsto k_{t}$ preserves this
structure, or more precisely, $k_{t}$ should be of the same type
\begin{equation}
k_{t}(\eta)=e_{\lambda}(\rho_{t},\eta)=\prod_{x\in\eta}\rho_{t}(x),\quad\eta\in\Gamma_{0}.\label{eq:chaotic_preservation}
\end{equation}

\item[]{2.} Relation (\ref{eq:chaotic_preservation}) is known as the \emph{chaos
preservation property} of the Vlasov hierar\-chy. It turns out
that equation (\ref{eq:chaotic_preservation}) implies, in general,
a non-linear differential equation
\begin{equation}
\frac{\partial\rho_{t}(x)}{\partial t}=\vartheta(\rho_{t})(x),\quad x\in\mathbb{R}^{d},\label{eq:nl-diff-eq-density}
\end{equation}
for $\rho_{t}$,
which is called the \emph{Vlasov-type kinetic equation}.
\end{enumerate}
\end{rem}

\begin{rem}
In general, if one does not start with a Poisson measure, the solution will leave the space $\mathcal{M}^{1}(\Gamma)$. To have a
bigger class of initial measures, we may consider
the cone inside  $\mathcal{M}^{1}(\Gamma)$
generated by convex combinations of Poisson measures, denoted by $\mathbb{P}(\Gamma)$.
\end{rem}

We would now like to generalize the above general scheme to obtain
the analog of \emph{kinetic fractional statistical dynamics} (or equivalently
\emph{mesoscopic fractional statistical dynamics}). It would be tempting simply to
replace the usual time derivative in equation (\ref{eq:nl-diff-eq-density})
by a time fractional derivative. Because the equation (\ref{eq:nl-diff-eq-density})
in general is non-linear, it is then much harder to obtain a solution.
But more essential is the question of the meaning of such an equation.
A naive use of the fractional derivative in the Vlasov equation is
not justified by the microscopic dynamics and its scaling.
Our alternative approach to realizing this generalization is described in the following section.

\section{Fractional statistical dynamics}

\label{sec:mfsd}

The procedure of Section\ \ref{sec:general_scheme} is suitable for describing
non-Markov evolutions. More precisely, in the FPE (\ref{eq:FPe})
we change the usual time derivative by the Caputo--Djrbashian
fractional time derivative $\mathbb{D}_{t}^{\alpha}$ (CDfd for short)
and then study the corresponding fractional dynamics.

In order to proceed, we first have to define the CDfd. Let $f:\mathbb{R}_{+}\longrightarrow\mathbb{R}$
be given; then the CDfd of $f$ is given in the Laplace transform
domain by
\[
\big(\mathcal{L}\mathbb{D}_{t}^{\alpha}f\big)(s)=s^{\alpha}(\mathcal{L}f)(s)-s^{\alpha-1}f(0),\quad s>0,\quad \alpha\in(0,1],
\]
where $\mathcal{L}f$ denotes the Laplace transform of $f$
\[
(\mathcal{L}f)(s)=\int_{0}^{\infty}e^{-st}f(t)\,dt.
\]
Another possible representation of the CDfd is
\[
\big(\mathbb{D}_{t}^{\alpha}f\big)(t)=\frac{1}{\Gamma (1-\alpha)}\frac{d}{dt}\int_{0}^{t}\frac{f(\tau)-f(0)}{(t-\tau)^{\alpha}}\,d\tau,\quad0<\alpha<1.
\]
In case $f$ is absolutely continuous, we have
\[
\big(\mathbb{D}_{t}^{\alpha}f\big)(t)=\frac{1}{\Gamma (1-\alpha)}\int_{0}^{t}
\frac{f'(\tau)}{(t-\tau)^{\alpha}}\,d\tau,\quad 0<\alpha<1.
\]
The definition of the CDfd has natural extensions to vector-valued or
measure valued functions on $\mathbb{R}_{+}$. We refer to the monographs
\cite{Podlubny99} and \cite{KST2006} for more details
and references concerning the CDfd.

We will introduce the fractional statistical dynamics for a given
Markov generator $L$ by changing the time derivative in the FPE to the CDfd.
The resulting fractional Fokker--Planck dynamics (if it exists!) will
act in the space of states on $\Gamma$, i.e., it
will preserve probability measures on $\Gamma$. The fractional
Fokker--Planck equation
\[
\begin{cases}
\mathbb{D}_{t}^{\alpha}\mu_{t}^{\alpha} &
=L^{*}\mu_{t}^{\alpha},\\ \mu_{t}^{\alpha}|_{t=0} &
=\mu_{0}^{\alpha}.
\end{cases}\tag{FFPE}
\]
describes a dynamical system with memory in the space of measures
on $\Gamma$. The corres\-ponding evolution no longer has the semigroup
property. However, if the solution $\mu_{t}$ of equation (\ref{eq:FPe1})
exists, then the subordination principle (see \cite{Pruss12}, \cite{Baz01,Bazhlekova00}
and references therein) gives the solution of the equation FFPE,
namely
\begin{equation}
\mu_{t}^{\alpha}=\int_{0}^{\infty}\Phi_{\alpha}(\tau)\mu_{\tau t^{\alpha}}\,d\tau.\label{eq:subordination1}
\end{equation}
Here $\Phi_{\alpha}(z)$ is the Wright function
\[
\Phi_{\alpha}(z):=\sum_{n=0}^{\infty}\frac{(-z)^{n}}{n!\Gamma(-\alpha n+1-\alpha)},
\]
a probability density function in $\mathbb{R}_+$). It is known
(see, for example, \cite{Gorenflo1999} and
\cite{Mainardi_Mura_Pagnini_2010}) that
\[
\Phi_{\alpha}(t)\geq0,\quad t>0,\quad\int_{0}^{\infty}\Phi_{\alpha}(t)\,dt=1,
\]
and that the moments of $\Phi_{\alpha}$ are given by
\begin{equation}\label{Mmoments}
\int_{0}^{\infty}t^{\delta}\Phi_{\alpha}(t)\,dt=\frac{\Gamma(\delta+1)}{\Gamma(\alpha\delta+1)},\quad\delta>-1.
\end{equation}
Its Laplace transform is given by
\[
\int_{0}^{\infty}e^{-\tau t}\Phi_{\alpha}(\tau)\,d\tau=E_{\alpha}(-t),\quad t>0,
\]
where $E_{\alpha}$ is the Mittag--Leffler function (see \cite{KST2006}):
\[
E_{\alpha}(z):=\sum_{n=0}^{\infty}\frac{z^{n}}{\Gamma(\alpha n+1)}.
\]
An application of the subordination principle may be justified in many particular models where
the evolution of correlation functions may be constructed by means a $C_0$-semigroup in a proper
Banach space. In general, the subordination formula  may be considered as a rule for the transformation
of Markov dynamics to fractional ones.

It is easy to see that  $\mu_{t}^{\alpha}$ is a measure. Actually, positivity follows from
that fact that for any measurable set $A$ we have
\[
\mu_{t}^{\alpha}(A)=\int_{0}^{\infty}\Phi_{\alpha}(t,s)\mu_{s}(A)\,ds\geq0,
\]
since $\mu_{s}$ is a measure and $\Phi_{\alpha}$ is a pdf. The
$\sigma$-additivity property may be verified using the standard procedure.
The FFPE equation may be written in terms of time-dependent correlation
functions as an infinite system of evolution equations, the so-called
\emph{hierarchical chain}:
\[
\begin{cases}
\mathbb{D}_{t}^{\alpha}k_{\alpha,t}^{(n)} &
=(L^{\triangle}k_{\alpha,t})^{(n)},\\ k_{\alpha,t}^{(n)}|_{t=0} &
=k_{\alpha,0}^{(n)},\quad n\geq0.
\end{cases}
\]
The evolution of the correlation functions should also be given by the
subordination principle. More precisely, if the solution $k_{t}$
of equation (\ref{eq:vlasov_hierarchy}) exists, then we have
\[
k_{\alpha,t}=\int_{0}^{\infty}\Phi_{\alpha}(\tau)k_{\tau t^{\alpha}}\,ds.
\]

\section{ Fractional kinetics and Poisson flows}

As in the case of Markov statistical dynamics addressed above, we may consider Vlasov-type scaling in the framework of the FFPE.
We know that the kinetic statistical dynamics for a Poisson initial
state $\pi_{\rho_0}$ is given by a flow of Poisson measures
\[
\mathbb{R}_{+}\ni t\mapsto\mu_{t}=\pi_{\rho_{t}}\in\mathcal{M}^{1}(\Gamma),
\]
where $\rho_t$ is the solution to the corresponding Vlasov kinetic
equation. Then the fractional kinetic dynamics of states may be
defined as the subordination of this flow (see comments above).  Specifically,
for $0<\alpha<1$ we consider the subordinated flow
\[
\mu_{t}^{\alpha}:=\int_{0}^{\infty}\Phi_{\alpha}(\tau)\mu_{\tau t^{\alpha}}\,d\tau=\int_{0}^{\infty}\Phi_{\alpha}(\tau)\pi_{\rho_{\tau t^{\alpha}}}\,d\tau.
\]
The family of measures $ \mu_{t}^{\alpha}$ is
no longer
a Poisson flow. We would like to analyze the properties of these subordinated flows to distinguish the
effects of fractional evolution.
Note first that the density
of the fractional kinetic state is given by the formula
$$
\rho_{t}^{\alpha}(x) =\int_{0}^{\infty}\Phi_{\alpha}(\tau)\rho_{\tau t^{\alpha}} (x)\,d\tau.
$$
The latter is the subordination of the solution to the Vlasov equation and is not related to a
fractional Vlasov equation as it is expected in several heuristic considerations in physics.

It is reasonable to study the properties of subordinated flows from a more general point
of view when the evolution of densities $\rho_t(x)$ is not necessarily
related to a particular Vlasov-type kinetic  equation. Similar transformations of Poisson flows do appear due to completely
different
motivations in several applications.  See, for example,
\cite{Silva-Oliveira2012, OM15} and, for the
related fractional Poisson process, \cite{MR1910034, MR2007003, MGS04,
Mainardi-Gorenflo-Vivoli-07, Uchaikin2008, Meerschaert2011},
and references therein.

Below we will study certain properties of the resulting flows affected
by fractional dynamics.

\subsection{Front propagation for the density }

Let us consider a density evolution of the form
 $$
 \rho_{t}(x)= {1\!\!1}_{[-1-vt, 1+vt]}(x), \quad t\geq 0,\quad x\in \mathbb{R},
 $$
 where $v>0$ is the constant speed of the density front. The subordinated density has the following
representation
 $$
 \rho_{t}^{\alpha}(x) =\int_{0}^{\infty}\Phi_{\alpha}(\tau){1\!\!1}_{[-1-vt^{\alpha}\tau, 1+vt^{\alpha}\tau]}(x) \,d\tau,
 $$
 and for $|x|>1$
 $$
 \rho_{t}^{\alpha}(x) =\int_{A(x,t)}^{\infty}\Phi_{\alpha}(\tau)\, d\tau,
 $$
 where
 $$
 A(x,\tau)= \frac{|x|-1}{vt^{\alpha}}.
 $$
 We have  $ \rho_{t}^{\alpha}(x)\to 1, t\to \infty, |x|>1$ and $ \rho_{t}^{\alpha}(x)\to 0, x\to \infty, t\geq0$.
Consider

 $$
 \Psi_{\alpha}(s)= \int_{s}^{\infty}\Phi_{\alpha}(\tau)\, d\tau.
 $$
 Due to monotonicity we may find a
unique $s_{\alpha}$ s.t.\ $ \Psi_{\alpha}(s_\alpha)=1/2$.
 Define the front of $\rho_{t}^{\alpha}$ for given $t>0$ as
$x\in \mathbb{R}$, for which $\rho_{t}^{\alpha}(x)=1/2$.
The motion of the front is then given by the formula
 $$
 |x|=1+s_{\alpha}vt^{\alpha}.
 $$
The latter result means that in the subordinated dynamics the density will be
expanded sub-linearly
 and more slower for smaller $\alpha\in (0,1)$.

 \subsection{Intermittency for subordinated flows}

 Each measure from the flow $  \mu_{t}^{\alpha}$ defines a generalized random process on
 ${\mathbb{R}}^{d}$ given for $f\in C_0(\mathbb{R}^d)$ by
 $$
 X_f(\gamma)= \sum_{x\in \gamma} f(x),\quad \gamma\in \Gamma.
 $$
 Let us consider the corresponding moments
 $$
 m_t^p(f)= \int X_{f}^{p}\, d \mu_{t}^{\alpha},\quad p\geq 1.
 $$
 The notion of asymptotic intermittency is well understood for  regular random fields; see for example
\cite{Carmona1994, Carmona1995}. In the case of generalized random fields this notion may be formulated as  follows.

 \begin{definition}{\rm(}Intermittency via moments{\rm)}.
 The flow  $\mu_{t}^{\alpha}, t\geq 0$ has the asymptotic intermittency property if for any
 $0\leq f \in C_0(\mathbb{R}^d)$ and for
all $p_1,\dots, p_n\in \mathbb{N}$ with
 $p_1+\cdots +p_n =p$
one had
 $$
 \lim_{t \to \infty} \frac{m_t^p (f)}{m_t^{p_1} (f) \dots m_t^{p_n} (f)} =\infty.
 $$
 \end{definition}
 This property means that moments of the random field grow in time progressively with the order.
 In the case of random point processes the leading growth of moments is defined in terms of correlation functions of the corresponding orders.

This gives us the option of reformulating the definition of asymptotic
intermittency in terms more convenient for our purposes.

 \begin{definition}{\rm(}Intermittency via correlation functions{\rm)}.\label{def:Intermittency_corr}
  The flow  $\mu_{t}^{\alpha}, t\geq 0$ has the   \break   asymptotic intermittency property if for any
  $\eta\in \Gamma_0$ and its decomposition $\eta = \eta_1 \cup \dots \cup \eta_n$ in disjunct subsets for the correlation function $k_{\mu_{t}^{\alpha}}$, one has
  $$
  \lim_{t\to\infty} \frac{ k_{\mu_{t}^{\alpha}}(\eta)}{k_{\mu_{t}^{\alpha}} (\eta_1)\dots k_{\mu_{t}^{\alpha}}(\eta_n)} =\infty.
  $$
  \end{definition}
For  a detailed discussion of relations between different versions of the intermittency property for
  random point processes, see \cite{KochKon}.

 Let us consider  the dynamics of the  density given by  $\rho_{t}(x)=e^{\beta t^{\sigma}}$,
$\beta,\sigma>0$. The flow  of Poisson measures $\pi_{\rho_{t}}$ has, for each $t\geq 0$,
correlation functions $k_{\pi_{\rho_{t}}}^{(n)}(x_1,\dots, x_n)
= e^{\beta n t^{\sigma}}$.  Therefore the intermittency is absent.

\begin{theorem}\label{thm:subexpgrowth}
Let $0<\alpha<1$ be given. Consider the subordinated flow for the Poisson
flow introduced above,
\[
\mu_{t}^{\alpha}:=\int_{0}^{\infty}\Phi_{\alpha}(\tau)\pi_{\rho_{\tau t^{\alpha}}}\,d\tau.
\]
Assume $\sigma (1-\alpha) <1$. Then the flow $\mu_{t}^{\alpha}$ has the asymptotic intermittency property.
\end{theorem}

\begin{proof}
The $n$-th correlation function of $\mu_{t}^{\alpha}$
is given by
 $$
\begin{aligned}
& (\rho_{t}^{\alpha})^{(n)}(x_{1},\ldots,x_{n})  =  \int_{0}^{\infty}\Phi_{\alpha}(\tau)(\rho_{\tau t^{\alpha}}(x))^{n}
\,d\tau
\\
& \qquad =\int_{0}^{\infty}\Phi_{\alpha}(\tau)e^{n\beta t^{\sigma\alpha}\tau^{\sigma}}\,d\tau
 =  \sum_{k=0}^{\infty}\frac{(n\beta t^{\sigma\alpha})^{k}}{k!}\int_{0}^{\infty}\Phi_{\alpha}(\tau)\tau^{\sigma k}
 \,d\tau
 \\
& \qquad = \sum_{k=0}^{\infty}\frac{(n\beta t^{\sigma\alpha})^{k}}{k!}\frac{\Gamma(\sigma k+1)}{\Gamma(\sigma k\alpha+1)}
  =  \sum_{k=0}^{\infty}\frac{n^{k}z^{k}}{k!}\frac{\Gamma(\sigma k+1)}{\Gamma(\sigma k\alpha+1)},
\end{aligned}
 $$
where $z:=\beta t^{\sigma \alpha}$. It is known
\cite{Braaksma1964} that the series converges for all values of
$z$ if and only if $\sigma <1/(1-\alpha)$, and that $$
(\rho_{t}^{\alpha})^{(n)}(x_{1},\ldots,x_{n})\sim
C(nz)^{1/(2\mu)}\exp \big(c(nz)^{1/\mu }\big), $$ where $C,c>0$,
$\mu =1+\sigma (\alpha -1)$. Since $0<\mu <1$, this asymptotic
behavior implies intermittency. In fact, due to
Definition~\ref{def:Intermittency_corr}, we need to consider the
limiting behavior of the ratio $$ \frac{\exp
(c(nz)^{\frac{1}{\mu}}}{\exp \sum_{k=1}^{m}
\big(c(n_kz)^{\frac{1}{\mu}}\big)},\quad  z=\beta t^{\sigma
\alpha} $$ for $t\to \infty$ under the assumption $\sum_{k=1}^{m}
n_k =n$. This limit is equal $+\infty$ due to the inequality $$
\left(\sum_{k=1}^{m} n_k\right)^{{\frac{1}{\mu}}}   >
\sum_{k=1}^{m} (n_k)^{{\frac{1}{\mu}}} $$ for $1/\mu > 1$ (see
\cite{BB}, Chapter 1, \S\,16).
\end{proof}

\subsection{Polynomially growing density}

Let us now consider the case of a polynomial density
 $\rho_{t}(x)=(1+t)^{p}$, $p\in\mathbb{N}$.  For any $n\in\mathbb{N}$,
the $n$th correlation function is given by
$$
\begin{aligned}
& (\rho_{t}^{\alpha})^{(n)}(x_{1},\ldots,x_{n})  =  \int_{0}^{\infty}\Phi_{\alpha}(\tau)
\rho_{\tau t^{\alpha}}(x_{1})\ldots\rho_{\tau t^{\alpha}}(x_{n})\,d\tau
\\
 & \qquad  =  \int_{0}^{\infty}\Phi_{\alpha}(\tau)(1+\tau t^{\alpha})^{pn}\,d\tau
  =  \sum_{j=0}^{np}\binom{np}{j}t^{\alpha j}\int_{0}^{\infty}\tau^{j}\Phi_{\alpha}(\tau)\,d\tau
  \\
 & \qquad  =  \sum_{j=0}^{np}\binom{np}{j}t^{\alpha j}\frac{\Gamma(j+1)}{\Gamma(j\alpha+1)}
  =  \sum_{j=0}^{np}\frac{(np)!}{(np-j)!}t^{\alpha j}\frac{1}{\Gamma(\alpha j+1)}
  \\
 & \qquad  =  \frac{(np)!}{\Gamma(\alpha np+1)}t^{\alpha np}+o(t^{\alpha np}).
\end{aligned}
$$
In particular, for $n=1$, the 1st correlation function is equal to
\[
(\rho_{t}^{\alpha})^{(1)}=\rho_{t}^{\alpha}=\frac{p!}{\Gamma(\alpha p+1)}t^{\alpha p}+o(t^{\alpha p}).
\]
Therefore we obtain
%
\begin{eqnarray*}
\frac{(\rho_{t}^{\alpha})^{(n)}}{(\rho_{t}^{\alpha})^{n}} & = & \frac{(np)!}{\Gamma(\alpha np+1)}t^{\alpha np}\times\left(\frac{\Gamma(\alpha p+1)}{p!}\right)^{n}\frac{1}{t^{\alpha pn}}+o(1)\\
 & = & \frac{(np)!}{\Gamma(\alpha np+1)}\times\left(\frac{\Gamma(\alpha p+1)}{p!}\right)^{n}+o(1),
\end{eqnarray*}
which is constant as $t$ goes to infinity. In conclusion, the power
growth of the 1st correlation function is not sufficient to organize
intermittency in the subordinated flow. Summarizing the above considerations, we conclude that subordinating
a flow (which corresponds to the dynamics of a system without intermittency)
is a way to {\emph{organize intermittency}}.

\section{Examples}

In this section we apply the general scheme of the  fractional
statistical dynamics developed here to concrete models, namely the \emph{contact model} and the \emph{pure birth model},
also known as the \emph{Surgailis pure birth model}.

\begin{example}{\rm(}Surgailis pure birth model{\rm)}. \label{exa:Surgailis_pbm}
This is an example in which the kinetic fractional statistical dynamics is a mixture of Poisson measures.
The Surgailis pure birth model (zero mortality) has generator given by
\[
(LF)(\gamma)=z\int_{\mathbb{R}^{d}}[F(\gamma\cup x)-F(\gamma)]\,dx.
\]
(cf.\ \cite{Kondratiev2008}). Starting from from the Poisson initial distribution $\mu_{0}=\pi_{\rho_{0}}$, the solution of the FPE
\[
\begin{cases}
\frac{\partial\mu_{t}}{\partial t} & =L^{*}\mu_{t},\\
\mu_{t}|_{t=0} & =\pi_{\rho_{0}}
\end{cases}
\]
is of the same type
\[
\mu_{t}=\pi_{zt+\rho_{0}}.
\]
The
solution of the fractional FPE
\[
\begin{cases}
\mathbb{D}_{t}^{\alpha}\nu_{t}^{\alpha} &
=L^{*}\nu_{t}^{\alpha},\\ \nu_{t}^{\alpha}|_{t=0} &
=\pi_{\rho_{0}}
\end{cases}
\]
is then given by the subordination principle as
\[
\nu_{t}^{\alpha}=\int_{0}^{\infty}\Phi_{\alpha}(s)\pi_{zt^{\alpha}s+\rho_{0}}\,ds.
\]
Hence the solution $\nu_{t}^{\alpha}$, $t>0$ is a mixture of Poisson
measures. The correlation function of the Poisson measure $\pi_{zt^{\alpha}s+\rho_{0}}$
is $((zt^{\alpha}s+\rho_{0})^{n})_{n=0}^{\infty}$, and
therefore the correlation function of the mixture $\nu_{t}^{\alpha}$ is, for $n\geq0$,
\begin{eqnarray*}
r_{t,\alpha}^{(n)} & = & \int_{0}^{\infty}\Phi_{\alpha}(s)(zt^{\alpha}s+\rho_{0})^{n}\,ds\\
 & = & \sum_{j=0}^{n}\binom{n}{j}(zt^{\alpha})^{j}\rho_{0}^{n-j}\int_{0}^{\infty}\Phi_{\alpha}(s)s^{j}\,ds.
\end{eqnarray*}
The absolute moments of $\Phi_{\alpha}$ (cf.\ eq.\ \eqref{Mmoments}) satisfy
\[
\int_{0}^{\infty}s^{j}\Phi_{\alpha}(s)\,ds=\frac{\Gamma(j+1)}{\Gamma(\alpha j+1)},\quad j>-1.
\]
Accordingly, the $n$th order correlation function of the measure $\nu_{t}^{\alpha}$ reduces to
\[
r_{t,\alpha}^{(n)}=\sum_{j=0}^{n}\binom{n}{j}\rho_{0}^{n-j}(zt^{\alpha})^{j}\frac{j!}{\Gamma(\alpha j+1)}=(zt^{\alpha})^{n}\frac{n!}{\Gamma(\alpha n+1)}+o\big((zt^{\alpha})^{n}\big).
\]

In particular
\[
\big(r_{t,\alpha}^{(1)}\big)^{n}=\left(\frac{zt^{\alpha}}{\Gamma(\alpha+1)}\right)^{n}+o\big((zt^{\alpha})^{n}\big),
\]
and thus
\begin{eqnarray*}
\frac{r_{t,\alpha}^{(n)}}{\big(r_{t,\alpha}^{(1)}\big)^{n}} & = & (zt^{\alpha})^{n}\frac{n!}{\Gamma(n\alpha+1)}\times\left(\frac{\Gamma(\alpha+1)}{zt^{\alpha}}\right)^{n}+o((zt^{\alpha})^{n})\\
 & = & \frac{(n-1)!}{\alpha\Gamma(n\alpha)}(\Gamma(\alpha+1))^{n}+o((zt^{\alpha})^{n}).
\end{eqnarray*}
From this we see that as $t\rightarrow\infty$ the above coefficient does
not explode, which tell us that this model has no asymptotic intermittency. In
other words, the power growth of the correlation function corresponding
to the FPE is not sufficient to realize asymptotic intermittency of the kinetic
fractional statistical dynamics. In the next example we show that
under strong growth on the $n$th order correlation function of the
FPE (exponential growth), the kinetic fractional statistical
dynamics does exhibit asymptotic intermittency.
\end{example}

\begin{example}{\rm(}Contact model{\rm)}.
\label{exa:contact-model}The contact model is one of the simplest
models in the theory of IPS.  Nevertheless, it has interesting properties,
e.g., its asymptotic behavior and the structure of its equilibrium measures.  We refer to \cite{FKK10} for more details.

The generator $L$ of the stochastic dynamics is given informally by
\[
(LF)(\gamma)=\sum_{x\in\gamma}m\big(F(\gamma\backslash x)-F(\gamma)\big)+\int_{\mathbb{R}^{d}}b(x,\gamma)\big(F(\gamma\cup x)-F(\gamma)\big)\,dx.
\]
Here $m>0$ is a constant mortality rate and the birth rate is
$$
b(x,\gamma) = \sum_{y\in \gamma} a(x-y),
$$
where $0\leq a \in L^1(\mathbb{R}^d)$ is even.

In the kinetic limit, the correlation functions of the contact model in the super critical regime are given by
\[
r_{t}^{(n)}(x_{1},\ldots,x_{n})=C^{n}e^{\beta nt}
\]
for certain $C,\beta >0$ \cite{Kondratiev2008}, \cite{FKK10}.
The correlation functions of the solution for the fractional kinetic dynamics are
then given by
\[
r_{t,\alpha}^{(n)}=C^{n}\int_{0}^{\infty}\Phi_{\alpha}(s)e^{\beta nt^{\alpha}s}\,ds=C^{n}E_{\alpha}(\beta nt^{\alpha}),\quad n\in\mathbb{N}.
\]
Using the asymptotic behavior of the Mittag-Leffler function $E_{\alpha}$
as $t\rightarrow\infty$ (see eq.\ (6.4) in \cite{Haubold2011}),
we can conclude that the kinetic fractional statistical dynamics
in the contact model does exhibit asymptotic intermittency. Of course, this statement is a particular case
of Theorem~\ref{thm:subexpgrowth} for $\sigma =1$.
\end{example}

{\it{Acknowledgments}}. We would like to thank the financial support of the DFG through  SFB 701, Bielefeld University
and the project I\&D: UID/MAT/04674/2013.

\end{document}